\documentclass[runningheads]{llncs}
\usepackage{xcolor}
\usepackage{graphicx}
\usepackage[linesnumbered,ruled,vlined]{algorithm2e}
\usepackage[normalem]{ulem}
\usepackage{float}
\usepackage{rotating}
\newcommand{\remove}[1]{}
\usepackage{nth}
\let\oldnl\nl
\newcommand{\nonl}{\renewcommand{\nl}{\let\nl\oldnl}}

\newenvironment{sloppypar*}
 {\sloppy\ignorespaces}
 {\par}

\newtheorem{innercustomlemma}{Lemma}
\newenvironment{customlemma}[1]
  {\renewcommand\theinnercustomlemma{#1}\innercustomlemma}
  {\endinnercustomlemma}
  
\newtheorem{innercustomthm}{Theorem}
\newenvironment{customthm}[1]
  {\renewcommand\theinnercustomthm{#1}\innercustomthm}
  {\endinnercustomthm}
\begin{document}
\title{Local Deal-Agreement Based Monotonic \\Distributed Algorithms for Load Balancing \\in General Graphs\\ \normalfont \normalsize (Preliminary Version)}
\author{}
\institute{}
\titlerunning{Distributed Algorithms for Load Balancing}
%
\author{Yefim Dinitz\inst{1} \and
Shlomi Dolev\inst{1} \and
Manish Kumar\inst{1}}

\authorrunning{Y. Dinitz et al.}

\institute{
Ben-Gurion University of the Negev, Israel\\
\email{\{dinitz, dolev\}@cs.bgu.ac.il, manishk@post.bgu.ac.il}
}

\maketitle              

\begin{abstract}
In computer networks, participants may cooperate in processing tasks, so that loads are balanced among them. We present local distributed algorithms that (repeatedly) use local imbalance criteria to transfer loads concurrently across the participants of the system, iterating until all loads are balanced. Our algorithms are based on a short local deal-agreement communication of proposal/deal, based on the neighborhood loads. They converge monotonically, always providing a better state as the execution progresses. Besides, our algorithms avoid making loads temporarily negative. Thus, they may be considered \emph{anytime} ones, in the sense that they can be stopped at any time during the execution. We show that our synchronous load balancing algorithms achieve \emph{$\epsilon$-Balanced} state for the continuous setting and \emph{1-Balanced} state for the discrete setting in all graphs, within $O(n D \log(n K/\epsilon))$ and $O(n D  \log(n K/D) + n D^2)$ time, respectively, where $n$ is the number of nodes, $K$ is the initial discrepancy, $D$ is the graph diameter, and $\epsilon$ is the final discrepancy. Our other monotonic synchronous and asynchronous algorithms for the discrete setting are generalizations of the first presented algorithms, where load balancing is performed concurrently with more than one neighbor. These algorithms arrive at a \emph{1-Balanced} state in time $O(n K^2)$ in general graphs, but have a potential to be faster as the loads are balanced among all neighbors, rather than with only one; we describe a scenario that demonstrates the potential for a fast ($O(1)$) convergence. Our asynchronous algorithm avoids the need to wait for the slowest participants' activity prior to making the next load balancing steps as synchronous settings restrict. We also introduce a self-stabilizing version of our asynchronous algorithm. 

\end{abstract}
\section{Introduction}
\label{sec:introduction} 

The load balancing problem is defined when there is an undirected network (graph) of computers (nodes), each one assigned a non-negative working load, and they like to balance their loads. If nodes $u$ and $v$ are connected by an edge in the graph, then any part of the load of $u$ may be transferred over that edge from $u$ to $v$, and similarly from $v$ to $u$.
The possibility of balancing loads among computers
is a fundamental benefit of the possible collaboration and coordination in distributed systems. The application and scope change over time, for example, the scopes include grid computing, clusters, and clouds. 

We assume that computation in the distributed systems is concurrent, that is, at any node, the only way to get information on the graph is by communicating with its neighbors. The initial information at a node is its load and the list of edges connecting it to its neighbors. The ideal goal of load balancing to make all nodes load equal to the average load; usually, some approximation to the ideal balancing is looked for, in order to avoid global operations (such as a preprocessed leader election based load balancing, where a leader node should be able to collect, while the system waits, store and process information from the entire system). An accepted global measure for the deviation of a current state from the balanced one is its (global) \emph{discrepancy}, defined as 
$K=L_{max}-L_{min}$, where $L_{max}$, resp., $L_{min}$, is 
the currently maximum, resp., minimum, node load in the graph. An alternative way to measure the deviation is in terms of local imbalance which is the maximal difference of loads between neighboring nodes.  A state is said to be \emph{$\epsilon$-Balanced} if that maximal local difference is at most $\epsilon$ (e.g., ``locally optimal load balancing'' in \cite{DBLP:conf/wdag/FeuilloleyHS15} is a 1-Balanced state in our terms). 
The continuous and discrete (integer) versions of the problem are distinguished. In the continuous setting, any amount of the load may be transferred over edges, while at the discrete setting, all loads and thus also all transfer amounts should be integers. 
In this paper, we concentrate on deterministic algorithms solving the problem in a time polynomial in the global input size, that is in the number $n$ of graph nodes and in the logarithm of the maximal load. 

The research on the load balancing problem has a long history. The two pioneering papers solving the problem are those of Cybenko \cite{DBLP:journals/jpdc/Cybenko89} and Boillat \cite{DBLP:journals/concurrency/Boillat90}. Both are based on the concept of diffusion: at any synchronized round, every node divides a certain part of its load \emph{equally} among its neighbors, keeping the rest of its load for itself; for regular graphs, the load fraction kept for itself is usually set to be equal to that sent to each neighbor. Many further solutions are based on diffusion, see \cite{DBLP:conf/focs/RabaniSW98,10.5555/59912,DBLP:conf/podc/AkbariBS12,DBLP:conf/soda/BerenbrinkCFFS11,DBLP:conf/stoc/AielloAMR93,DBLP:conf/stoc/GhoshLMMPRRTZ95}. Most of the papers consider \emph{$d$-regular} graphs, and use Markov chains and ergodic theory for deriving the rate of convergence. This approach works smoothly in the continuous setting. 

Diffusion methods in the discrete setting requires rounding of every transferred amount, which makes the analysis harder; Rabani et al.\ \cite{DBLP:conf/focs/RabaniSW98} made substantial advancement in that direction.
However, the final discrepancy achieved by their method is not constant. As mentioned in \cite{DBLP:conf/podc/ElsasserS10}, the discrepancy cannot be reduced below $\Omega(d_{min} \cdot D)$ by a deterministic diffusion-based algorithm, where $D$ is the diameter of the graph and $d_{min}$ is the minimum degree in the graph. The algorithms in papers  \cite{DBLP:journals/jgaa/ElsasserMS06,DBLP:conf/podc/ElsasserS10} are based on \emph{randomized} post-processing, using random walks, for decreasing the discrepancy to a constant with high probability. 

One of the alternatives to diffusion is the \emph{matching} approach. There, a node matching is chosen at the beginning of each synchronized round, and after that every two matched nodes balance their loads. See, e.g., \cite{DBLP:conf/stoc/FriedrichS09,DBLP:conf/focs/SauerwaldS12}. In known deterministic algorithms, those matching, are usually chosen in advance according to the graph structure; in randomized algorithms, they are chosen randomly. The deterministic matching algorithms of Feuilloley et al.\ \cite{DBLP:conf/wdag/FeuilloleyHS15} achieve a 1-Balanced final state for general graphs, in time not depending on the graph size $n$ and depending cubically, and also exponentially for the discrete setting, on the initial discrepancy $K$. These algorithms extensively use communication between nodes. 
(In the further discussion, we do not relate to the approach and results of Feuilloley et al.\ \cite{DBLP:conf/wdag/FeuilloleyHS15}, since the used computational model is different from that we concentrate on.)
Yet another \emph{balancing circuits} \cite{DBLP:conf/stoc/AspnesHS91,DBLP:conf/focs/RabaniSW98} approach is based on sorting circuits \cite{DBLP:conf/focs/RabaniSW98}, where each 2-input comparator balances the loads of its input nodes (instead of the comparison), see e.g., \cite{10.1145/103418.103421}. Randomized matching and balancing circuits algorithms of \cite{DBLP:conf/focs/SauerwaldS12} achieve a constant final discrepancy w.h.p.

All the above-mentioned papers use the synchronous distributed model. The research on 
load balancing in the \emph{asynchronous} distributed setting (where the time of message delivery 
is not constant and might be unpredictably large) is not extensive. The only theoretically based approach suggested for it is turning the asynchronous setting into synchronous by appropriately enlarging the time unit, see e.g., \cite{DBLP:conf/stoc/AielloAMR93}.

Note that the suggested solutions ignore the possibility of a short coordinating between nodes.
A node just informs each of its neighbor on the load amount transferred to it, if any, and lets the neighbors know on its own resulting load after each round. 
In previous diffusion based load balancing works, only its own load and the pre-given distribution policy define the decisions at each node; as a maximum, a node also pays attention to
the current loads at its neighbors. 

The above-mentioned approaches leave the following gaps in the suggested solutions (among others):
\begin{itemize}

    \remove{
    \item  \sout{While the convergence rate to the desired state is analyzed, \sout{there is no guarantee} \textcolor{red}{ we look for a combination of guarantees and, no algorithm satisfies all of them,} as to the quality of graph states along the way. In particular, at intermediate states, either loads might become negative (see, e.g, \cite{DBLP:conf/podc/ElsasserS10,DBLP:conf/soda/FriedrichGS10,DBLP:conf/focs/SauerwaldS12,DBLP:conf/podc/AkbariBS12}), which is not appropriate in the output, or the discrepancy might grow w.r.t.\ previous states (see, e.g, \cite{DBLP:conf/focs/RabaniSW98,DBLP:conf/focs/SauerwaldS12}). Following \cite{DBLP:conf/aaai/DeanB88,DBLP:conf/uai/Horvitz87}, we call an algorithm \emph{anytime}, \textcolor{red}{if each of its intermediate solution is not worse than the previous intermediate solution.} In particular, this property of the algorithm is important for the message-passing settings where the actual time of message delivery might be long. \textcolor{blue}{The shortcomings, as shown above}, prevent the known load balancing algorithms from being anytime ones.}
    }
    
    \item
    Following \cite{DBLP:conf/aaai/DeanB88,DBLP:conf/uai/Horvitz87}, we call an algorithm \emph{anytime}, if any of its intermediate solutions is feasible and not worse than the previous ones; thus, it may be used immediately in the case of a need/emergency. In particular, this property of algorithms is important for the message-passing settings where the actual time of message delivery might be long.
    In the suggested solutions to the load balancing problem, to the contrast, either, loads might become negative (see, e.g, \cite{DBLP:conf/podc/ElsasserS10,DBLP:conf/soda/FriedrichGS10,DBLP:conf/focs/SauerwaldS12,DBLP:conf/podc/AkbariBS12}), which is not appropriate in an output, or the discrepancy might grow w.r.t.\ previous states (see, e.g, \cite{DBLP:conf/focs/RabaniSW98,DBLP:conf/focs/SauerwaldS12}).
    
    \item  
    The \emph{general graph} case is more or less neglected in the literature; instead, most of the papers focus on $d$-regular graphs. There are only a few brief remarks on the possibility of generalizing the results to general graphs, see, e.g., \cite [second paragraph] {DBLP:conf/focs/RabaniSW98}.

    \remove{
    \sout{\item
    There is a factor of $1/(1 - \lambda)$ in almost all time bounds, where $\lambda$ is the second largest eigenvalue of some related diffusion matrix M. \textcolor{blue}{The double gap} is in: (\textit{a}) the uncertainty of M, and (\textit{b}) the lacking evaluation of $(1 - \lambda)$, in the terms of the given general graph. }
    %
    %
    \sout{\item  Even the \textcolor{blue}{most advanced (why?)} deterministic algorithm for the discrete problem version of Rabani et al.\ \cite{DBLP:conf/focs/RabaniSW98} has two shortcomings:}}
    \item  The \emph{final discrepancy} achieved by the deterministic algorithms is far not a constant.
    
    \remove{
    \sout{algorithm is  \textcolor{blue}{not so small} \textcolor{red}{(R3: Be precise or at least say that it is large in comparison with something else if the precise number is hard to state. )}. No discussion is given on decreasing it to the minimal possible value.}
     \sout{   \item  The convergence rate is provided in terms of eigenvalues of some (\textcolor{blue}{in general, uncertain}) matrix. Neither the best choice of a matrix making the convergence rate maximal for the given graph nor the worst-case bound for the general graph is discussed.}}
    \item  The \emph{asynchronous setting} is explored quite weakly in the literature. There, no worst-case bounds for achieving the (almost) balanced state are provided.
\end{itemize}

\enlargethispage{\baselineskip}
We suggest using the distributed computing approach in load balancing, based on \emph{short agreement between neighboring nodes}. By our algorithms, we tried 
to close the gaps as above, as far as possible.
In this paper, we concentrate on deterministic algorithms with deterministic worst-case analysis, leaving the development of randomized algorithms (that fit scenarios in which no hard deadlines are required) for further research.

We develop \emph{local} distributed algorithms, in the sense that each node uses the information on its neighborhood only, with no global information collected at the nodes. The advantage to this is that there is no ``entrance fee'', so that the actual running time of an algorithm can be quite small, if we are lucky, say, if we have a starting imbalance that can be balanced in small areas without influencing the entire system; in Section~\ref{s:diffusion}, we describe a scenario, which demonstrates the potential for a fast convergence.
Accordingly, we choose the ideal goal of a discrete local algorithm to be a \emph{$1$-Balanced state}, rather than a small discrepancy. Indeed, let us consider a graph that is a path of length $2n-1$, where the node loads are 0, 1, 1, 2, 2,  ..., $n-1$, $n-1$, $n$, along it. No node has a possibility to improve the load balancing in its neighborhood, while the discrepancy is $n$,
a half of the number of graph nodes.

We say that a load balancing algorithm is \emph{monotonic} if the following conditions always hold during any its execution: (a) each load transfer is from a higher loaded node to a less loaded one, and (b) the maximal load value never increases and the minimal load value never decreases. By item (b), the discrepancy never increases during any execution of a monotonic algorithm, and negative loads are never created, that is, \emph{any monotonic algorithm is anytime}.

Our main results are as follows, where 
$D$ is the graph diameter, and $\epsilon$ is the 
an arbitrarily small bound 
for the final discrepancy bound.
We call an algorithm a \emph{single proposal} one, if each node at each round proposes a load transfer to at most one node and accepts a proposal of at most one node, and a \emph{distributed proposal} one, when proposals can be made/accepted to/from several nodes.
\begin{itemize}
    \item 
    In the continuous setting, the first synchronized deterministic algorithm for \emph{general graphs}, which is \emph{monotonic} and works in time $O(n D \log(n K/\epsilon))$.
    \item 
    In the discrete setting, the first deterministic algorithms for \emph{general graphs} achieving a \emph{1-Balanced state} in time depending on the initial discrepancy logarithmically. It is \emph{monotonic} and works in time $O(n D  \log(n K/D) + n D^2)$.
    \item
    Distributed proposals scheme that generalizes the single proposal scheme of those two algorithms.
    \item  
    First \emph{asynchronous anytime} load balancing algorithm.
    \item
    A \emph{self-stabilizing} version of our asynchronous algorithm.
\end{itemize}
\emph{Remark\/}:
Let us consider the place of single proposal algorithms 
in the classification of load balancing algorithms. At each round, each node $u$ participates in at most two load transfers: one where it transfers a load and one where it receives a load. In other words, the set of edges with transfers, each oriented in the transfer direction, forms \emph{a kind of matching}, where both in- and out-degree of any node are at most 1 each. This property may be considered a generalization of the matching approach to load balancing. Therefore, one may call single proposal algorithms ``short agreement based generalized matching'' ones. 

Let us compare the results achieved by our two single proposal algorithms with those achieved by deterministic algorithms in Rabani et al.\ \cite{DBLP:conf/focs/RabaniSW98}. As of the scope of \cite{DBLP:conf/focs/RabaniSW98}, details are provided there for the $d$-regular graphs only, and there is no anytime property. For continuous setting, consider the time bound $O\left(\frac{\ln (Kn^2/\epsilon)}{(1-\lambda)}\right)$ of \cite{DBLP:conf/focs/RabaniSW98} for reaching the discrepancy of $\epsilon$ in the worst case. It is the same as our bound in its logarithmic part. Its factor of $\frac 1{1-\lambda}$ is $\Theta(n^2)$ in the worst case, e.g., for a graph that is a cycle. Also, the factor $nD$ in our bound is $\Theta(n^2)$ in the worst case. We conclude that 
the worst-case bounds overall graphs of the same size are both $O(n^2 \log (Kn))$. Note that the class of instances bad for the time bound $O \left( \frac{\log (Kn)}{1-\lambda}\right)$ is quite wide. Consider an \emph{arbitrary} graph $G$ and add to it a cycle of length $n$ with a single node common with $G$. We believe that the time bound for the resulting graph will not be better than for a cycle of length $n$, that is $O(n^2 \log (Kn))$.
For the discrete setting, we achieve a 1-Balanced state, which is not achieved in Rabani et al.\  \cite{DBLP:conf/focs/RabaniSW98}. 
As on discrepancy (which is our secondary goal), we guarantee  discrepancy $D$, while the final discrepancy in \cite{DBLP:conf/focs/RabaniSW98} is $O\left( \frac{d \log n}{1-\lambda} \right)$;
we believe that in the class of instances as described above but with the added cycle of length $D$, the latter bound is $O(d D^2 \log(n))$.

\remove{
\footnote{\textcolor{red}{
   We found it difficult to include parameter $\lambda$ into comparison, since it relates to the diffusion matrix, which is not well-defined by a general graph itself.}
   }
   }
\remove{
\sout{
It may be of some technical interest to compare the time bound $O(nD \log(nK))$ of the first phase of our algorithm, achieving discrepancy $2D$, with the time bound $O \left( \frac{\log (Kn)}{1-\lambda}\right)$ of \cite{DBLP:conf/focs/RabaniSW98}, achieving a \textcolor{blue}{\emph{much worse} (R3: bit aggressive and it depends on the range 
that you are looking for in terms of K versus n (I guess here you don't want 
bounded degree d). If K is huge, then your algorithm is worse. And I guess you 
consider quite large K, otherwise you wouldn't be shocked by the $K^3$ in [12])}, in general, discrepancy $O\left( \frac{d \log n}{1-\lambda} \right)$. These bounds look similar, with the same logarithmic part and the factor of $\Theta(n^2)$ in the worst case over all graphs of the same size. }}

Summarizing, the main contribution of this paper to the load balancing research is turning attention of the load balancing community: (a) to using standard distributed computing methods, that is to short agreement based local distributed algorithms, (b) to the general graphs case, and (c) to anytime monotonic algorithms, and making first steps in these three directions.

\enlargethispage{\baselineskip}

\textbf{Organization of the paper.} Section~\ref{s:relatedwork} presents a summary of related work. The two single proposal algorithms are presented in Section~\ref{s:concentrated}, with the analysis. Section~\ref{s:diffusion} describes and analyzes the distributed proposal algorithms. Section~\ref{s:async} discusses the asynchronous version of load balancing algorithms. Section~\ref{s:selfstab} introduces a self-stabilizing version of our asynchronous algorithm. The conclusion appears in Section~\ref{s:conclusion}.

\subsection{Diffusion vs Deal-Agreement based Load balancing Algorithms}
To gain intuition we start with an example. Consider Fig. 1, each node is transferring the load to its neighbor but in the next iteration node $A$ has a large load, which may lead to more iterations of the algorithm for balancing the load. Note that the original load, 5, becomes 21 following the load diffusion. In each round of deal-agreement based load balancing algorithm, each node computes the \emph{deal} (amount of load accepted by a node without violating the value of \textit{TentativeLoad}) individually then transfers the load to its neighbors. An example of our deal-agreement based load balancing appears in Fig. 2, node $A$ receives the load from the neighbors according to the proposal up to the value of \textit{TentativeLoad} but never exceeds the \textit{TentativeLoad} (load of the node after giving loads to its neighbors). Note that the loads are transferring towards \emph{1-Balanced} without violating the condition of monotonicity. Our deal-agreement based algorithm  is an anytime and monotonic algorithm that ensures the maximal load bound of nodes is not exceeded. Namely, if there is an upper bound on the loads that can be accumulated in a single node, then starting with loads that respect this bound the bound is not violated due to load exchanges. As opposite to diffusion based algorithms, e.g., load value 21 on node A in Fig. 1 may exceed the capacity of a node.

\begin{figure*}
\centering 
\includegraphics[page=8,width=0.95\textwidth]{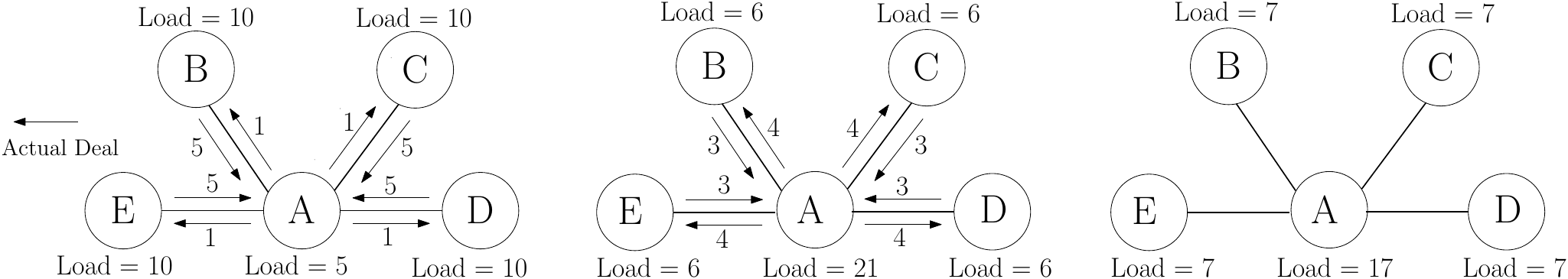}
\caption{Example of Diffusion Load Balancing Algorithm }
\end{figure*}

\begin{figure*}
\centering
\includegraphics[page=9,width=0.95\textwidth]{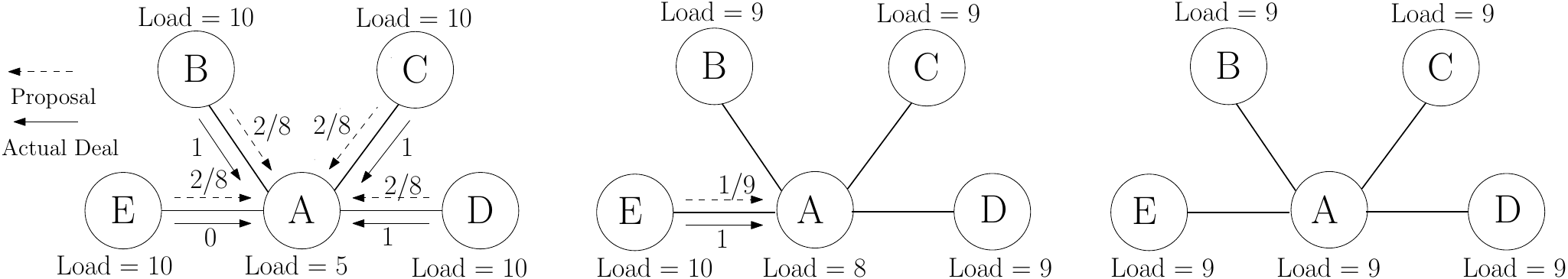}
\caption{Example of Deal-Agreement Based Load Balancing Algorithm}
\end{figure*}

\subsection{Related Work}
\label{s:relatedwork}
\remove{
\textcolor{red}{Token Distribution Problem \cite{DBLP:conf/focs/PelegU86,DBLP:journals/algorithmica/HeideOW96} is a variation of static load balancing problem, in which each node holds some arbitrary tokens and the node must exchange tokens as rapidly as possible so that they all end up with approximately the same number of tokens on each node.}}

In this section, we review some related results for load balancing dealing with diffusion, matching, and random-walk based algorithms only.

\begin{itemize}
    \item \textbf{Continuous Diffusion. }Continuous load balancing is an ``ideal'' case (continuous case of Rabani et al.\ \cite{DBLP:conf/focs/RabaniSW98}) where the load can be divided arbitrarily. Therefore, it is possible to balance the load perfectly. Muthukrishnan et al.\ \cite{DBLP:conf/spaa/GhoshMS96} refers the above continuous diffusion model discussed in Rabani et al.\ \cite{DBLP:conf/focs/RabaniSW98}, as the \textit{first order scheme} and extend to \textit{second order scheme} which take $O(\log (Kn))/(1-\lambda))$ time for achieving discrepancy of $O(dn/(1-\lambda))$, where $\lambda$ is the second largest eigenvalue of diffusion matrix, $K$ is initial discrepancy and $n$ is number of node in the graph. Further progress was made by  Rabani et al.\ \cite{DBLP:conf/focs/RabaniSW98} who introduced the so-called \emph{local divergence}, which is an effective way to compute the deviation between actual load and deviation generated by a Markov chain. They proved that the local divergence yields bound on the maximum deviation between the continuous and discrete case for both the diffusion and matching model. They proved that local divergence can be reduced to $O (d  \log (Kn)/(1 - \lambda))$ in $O(\log (Kn))/(1-\lambda))$ rounds for a \textit{d}-regular graph.
    
    \item \textbf{Discrete Diffusion. }In discrete load balancing only non-divisible loads are allowed to transfer. In this case, the graph can not be completely balanced. Akbari et al.\ \cite{DBLP:conf/podc/AkbariBS12} discussed a randomized and discrete load balancing algorithm for general graph that balances the load up to a discrepancy of $O(\sqrt{d \log n})$ in $O (\log (Kn)/(1 - \lambda))$ time, where $d$ in maximum degree, $K$ is the initial discrepancy, and $\lambda$ is the second-largest eigenvalue of the diffusion matrix. Using the algorithm of Akbari et al.\ \cite{DBLP:conf/podc/AkbariBS12}, discrepancy reduces for other topology also. Such as, for hypercubes, discrepancy reduces to $O( \log n)$ and for expanders and torus graph reduces to $O(\sqrt{\log n})$.
    
 \end{itemize}

\enlargethispage{\baselineskip}

\textit{Randomized diffusion} based algorithms \cite{DBLP:conf/soda/BerenbrinkCFFS11,DBLP:conf/podc/AkbariBS12,DBLP:conf/soda/FriedrichGS10} are algorithms in which
 every node distributes its load as evenly as possible among its neighbors and itself. If the remaining load is not possible to distribute without dividing some load then the node redistributes the remaining loads to its neighbors randomly. In Berenbrink et al.\ \cite{DBLP:conf/soda/BerenbrinkCFFS11} authors show randomized diffusion based discrete load balancing algorithm for which discrepancy depends on the expansion property of the graph (e.g., $\lambda$ which is the second-largest eigenvalue of the diffusion matrix and $d$ maximum degree of any node in the graph). Other \textit{Quasi-random diffusion} based algorithm for general graph \cite{DBLP:conf/soda/FriedrichGS10} considered \textit{bounded-error property}, where sum of rounding errors on each edge is bounded by some constant time at all time. The discrepancy results of the randomized algorithms in  \cite{DBLP:conf/soda/BerenbrinkCFFS11,DBLP:conf/soda/FriedrichGS10} are further improved to $O(d^2 \sqrt{\log n})$ and $O(d \sqrt{\log n})$ respectively by applying the results from \cite{DBLP:conf/focs/SauerwaldS12}, where tighter bounds are obtained for certain graph parameters used in discrepancy bounds of \cite{DBLP:conf/soda/BerenbrinkCFFS11,DBLP:conf/soda/FriedrichGS10}.

One of the alternatives to diffusion is the \emph{matching} approach, also known as the dimension exchange model. In this model for each step, an arbitrary matching of nodes is given, and two matched nodes balance their loads. Friedrich et al.\ \cite{DBLP:conf/stoc/FriedrichS09} considered a dimension exchange algorithm for the matching model. There every node which is connected to a matching edge computes the load difference over that edge. If the load difference is positive, the load moves between the nodes of that edge. This work reduces the discrepancy to $O \left( \sqrt{ \log^3 n / (1- \lambda)} \right)$ in $O(\log (Kn)/ (1-\lambda))$ steps with high probability for general graph. This result \cite{DBLP:conf/stoc/FriedrichS09} is improved by Sauerwald et al.\ \cite{DBLP:conf/focs/SauerwaldS12} where they achieve constant  discrepancy in $O(\log(Kn)/(1-\lambda))$ steps with high probability for regular graph in the random matching model. The constant is independent of the graph and the discrepancy $K$. The random matching model is an alternative model of balancing circuit, where random matching generated in each round. Additionally, the deterministic matching algorithm of \cite{DBLP:conf/wdag/FeuilloleyHS15} achieves a constant final discrepancy of $O((\log n)^ \epsilon)$ for an arbitrarily small constant $\epsilon > 0$, but after a large number of rounds, depending cubically on the initial discrepancy where $K \geq n$. 

\begin{table}[!b]
\label{t:table} 
\caption{Comparison of different Load Balancing algorithms for the \emph{general graph} (\emph{where $n$ is the number of nodes, $d$ is the maximum degree of a node, $K$ is the initial discrepancy, $D$ is the graph diameter, $\lambda$ is the second highest eigenvalue of the diffusion matrix, $\alpha$ is the edge expansion value of the graph, $\epsilon > 0$ is an arbitrarily small constant, D' represent Deterministic algorithm, R' represent Randomized algorithm, I represent Discrete(Integer), C represent Continuous})}
\label{tab:commands}

\scalebox{0.715}{
\begin{tabular}{l*{6}{c}r}
\hline
 Algorithms & Type & Approach & Transfer & Final discrepancy & Rounds & Anytime  \\
(References) &  & (Diffusion/ & (I/C) & ($L_{max} - L_{min}$) & (Steps) &    \\
 & &Matching) & & & & \\
\hline
  \textit{Synchronous Model}& & & & & & \\
   \hline
    Rabani et al. \cite{DBLP:conf/focs/RabaniSW98}& D'  & Diffusion/Matching & I & $ O\left( \frac{d \log n}{1-\lambda} \right)  $ & $O \left( \frac{\log (Kn)}{1-\lambda}\right)$ & No   \\
    
     &  &  & C  & \emph{$\epsilon$} & $ O\left(\frac{\ln (Kn^2/\epsilon)}{(1-\lambda)}\right)$ & No   \\
    
\hline
Muthukrishnan  & D' & Diffusion & C & $O \left( \frac{ dn}{1 - \lambda}\right)$  & $O \left( \frac{ \log (Kn)}{1 - \lambda}\right)$ & No \\
et al. \cite{DBLP:conf/spaa/GhoshMS96} & & & & & & \\

\hline
    Akbari et al. \cite{DBLP:conf/podc/AkbariBS12} & R' & Diffusion & I &$O(\sqrt{d \log n})$ \emph{w.h.p}&$O \left( \frac{d  \log (Kn)}{1 - \lambda}\right)$ & No \\
    
    Berenbrink et al. \cite{DBLP:conf/soda/BerenbrinkCFFS11} & R' & Diffusion & I & $O \left(d \sqrt{\log n} +\sqrt{ \frac{d  \log n  \log d}{1 - \lambda}}\right)$ & $O \left( \frac{\log (Kn)}{1 - \lambda}\right)$ & No \\
    
    & & & & \emph{w.h.p} & & \\

    Friedrich et al. \cite{DBLP:conf/soda/FriedrichGS10} & R' & Diffusion & I & $ O\left( \frac{d \log \log n}{(1- \lambda)}  \right)  $ \emph{w.h.p} & $O \left( \frac{\log (Kn)}{1 - \lambda}\right)$  & No \\
\hline    
 Friedrich et al. \cite{DBLP:conf/stoc/FriedrichS09} & R' & Matching & I  & $ O\left( 
\sqrt{\frac{\log^3 n}{(1- \lambda)} }\right)$ \emph{w.h.p} & $O \left( \frac{\log (Kn)}{1 - \lambda}\right)$ & No \\   
  
 Sauerwald et al. \cite{DBLP:conf/focs/SauerwaldS12}& R' & Matching & I &   Constant \emph{w.h.p} &  $O \left( \frac{ \log (Kn)}{1 - \lambda}\right)$  & No\\  
\hline   
 Feuilloley et al. \cite{DBLP:conf/wdag/FeuilloleyHS15}& D' & Matching & I  & Constant & $O(K^3 d^K)$ & No \\
&  &  & C & Constant & $O(K^3 \log d)$ & No  \\   
      
\hline
   Elsässer et al. \cite{DBLP:journals/jgaa/ElsasserMS06}&R' &Diffusion+ & I  & Constant & $O\left((\log K) + \frac{(\log n)^2}{(1- \lambda)}\right)$  & No \\
& &Random Walk& & & \emph{w.h.p}&\\

Elsässer et al. \cite{DBLP:conf/podc/ElsasserS10}&R' & Diffusion+ & I  & Constant & $O \left( \frac{\log (Kn)}{1 - \lambda}\right)$ \emph{w.h.p} & No\\
&&Random Walk&&&&\\

Elsässer et al. \cite{DBLP:conf/podc/ElsasserS10}&R' & Diffusion+ & I  & Constant & $O(D \log n)$ \emph{w.h.p} & No\\
&&Random Walk&&&&\\

\hline
Our Algorithm \ref{algorithm:new_cont}& D'  & Deal-Agreement based  & C & $\epsilon$  & $O(nD \log(nK/\epsilon))$ & Yes \\
& & Generalized & & & & \\
& & Matching& & & & \\

Our Algorithm \ref{algorithm:new_disc}& D'  & Deal-Agreement based & I & 1-Balanced &  $O(n D  \log(n K) + n D^2)$ & Yes\\
& & Generalized & & & & \\
& & Matching& & & & \\

Our Algorithm \ref{algo:a1}& D'  & Deal-Agreement based  & I & 1-Balanced & $O(nK^2)$ & Yes \\
& & Diffusion& & & & \\
\hline

    \textit{Asynchronous Model}& & & & & & \\
   \hline
     \textit{Partially Async.:}& & & & & & \\
     J. Song \cite{DBLP:conf/ipps/Song93}& D'  & Diffusion &  C & $ \lceil D/2 \rceil $ &  -  & No \\
   
   \textit{Fully Asynchronous:}& & & & & & \\
   
Aiello et al. \cite{DBLP:conf/stoc/AielloAMR93}& D'  & Matching & I&$O \left( \frac{d^2 \log n}{\alpha}\right)$  & $O(K / \alpha)$ & No \\
Ghosh et al. \cite{DBLP:conf/stoc/GhoshLMMPRRTZ95}& D'  & Matching & I & $O \left( \frac{d^2 \log n}{\alpha}\right)$  & $O(K / \alpha)$ & No \\
Our Algorithm \ref{algo:a4}& D'  & Diffusion & I & 1-Balanced & $O(nK^2)$ & Yes\\
  \hline
 \textit{Self-Stabilizing Algo.}\\
  \hline
  Flatebo et al. \cite{504130} & D' & - & I & - & - & No\\
  \hline

\end{tabular}
}
\end{table}

Elsässer et al.\ \cite{DBLP:journals/jgaa/ElsasserMS06,DBLP:conf/podc/ElsasserS10} propose load balancing algorithms based on \emph{random walks}. This approach is more complicated than simple diffusion based load balancing. In these algorithms the final stage uses concurrent random walk to reduce the maximum individual load. Load transfer on an edge may be smaller or larger than \textit{(Load Difference)}/(\textit{degree + 1}). Elsässer et al.\ \cite{DBLP:journals/jgaa/ElsasserMS06} present an approach achieving a constant discrepancy after $O((\log K) + (\log n)^2(1- \lambda))$ steps, where $\lambda$ is the second largest eigenvalue of the diffusion matrix and $K$ is initial discrepancy. Elsässer et al. \cite{DBLP:conf/podc/ElsasserS10} presented a constant discrepancy algorithms which is an improvement of the algorithm of \cite{DBLP:journals/jgaa/ElsasserMS06} and takes $O \left( \frac{\log (Kn)}{1 - \lambda}\right)$ \emph{w.h.p.} time.

\emph{Asynchronous} load balancing algorithm, where local computations and messages may be delayed but each message is eventually delivered and each computation is eventually performed. Aiello et al.\ \cite{DBLP:conf/stoc/AielloAMR93} and Ghosh et al.\ \cite{DBLP:conf/stoc/GhoshLMMPRRTZ95} introduced matching based asynchronous load balancing algorithms with the restriction that in each round only one unit load transfers between two nodes. By this restriction, these asynchronous load balancing algorithms require more time to converge. These algorithms \cite{DBLP:conf/stoc/AielloAMR93,DBLP:conf/stoc/GhoshLMMPRRTZ95} suggest turning the asynchronous setting into synchronous by appropriately enlarging the time unit.

The literature on \emph{ Self-stabilizing Token Distribution Problem} \cite{DBLP:conf/opodis/SudoDLM18,DBLP:conf/sirocco/SudoDLM18}, where any number of tokens are arbitrarily distributed and each node has an arbitrary state. The goal is that each node has exactly $k$ tokens after the execution of algorithm. Sudo et al.\ \cite{DBLP:conf/opodis/SudoDLM18,DBLP:conf/sirocco/SudoDLM18} introduced the algorithm for asynchronous model and rooted tree networks, where root can push/pull tokens to/from the external store and each node knows the value of $k$.

Two self-stabilizing algorithms for transferring the load (task) around the network are presented by Flatebo et al.\ \cite{504130}. These algorithms are in terms of a new \emph{task received} from the environment that triggers \emph{send task} or \emph{start task}, rather than being activated when no task is received, which is our scope here.

A comparison of our results with existing results is given in Table \ref{tab:commands}. 

\section{
Single Proposal Load Balancing Algorithms}
\label{s:concentrated}

This section presents two \emph{single proposal} monotonic synchronous local algorithms for the load balancing problem and their analysis. The continuous version works in time $O(n D \log(n K/\epsilon))$, while the discrete version works in time $O(n D \log(n K/D) + n D^2)$.

\subsection{Continuous Algorithm and Analysis}

\begin{algorithm}
\caption{Single Proposal: Continuous Version}
\label{algorithm:new_cont}
\KwIn{An undirected graph $G = (V,E, load)$ } 

\KwOut{Graph $G$ with discrepancy at most $\epsilon$}

         \SetKwFunction{FMain}{}
         \SetKwProg{Fn} {Execute forever} { do}{}
         \DontPrintSemicolon
          \Fn{}{
                \For{every node $u$} { 
                \If{$u$ has at least one neighbor with a strictly smaller load}
                {
                Find the neighbor, $v$, with the maximal difference $load(u) - load(v)$ (break ties arbitrarily)\\
                $u$ sends to $v$ a transfer proposal of $p_{uv}=(load(u) - load(v))/2$\\
                }
                }
                \For{every node $u$} {
                \If{there is at least one transfer proposal to $u$}{
                Find a neighbor, $w$, proposing to $u$ the transfer of maximum value, $p_{wu}$\\
                Node $u$ makes a deal: increases its load by $p_{wu}$ and informs node $w$ on accepting its proposal\\
                }
                }
                \For{every node $u$} {
                Node $u$ updates its load w.r.t.\ the deal issued by it and the deal made on its proposal, if any, and sends the current value of $load(u)$ to every its neighbor\\
                }
            }

\end{algorithm}

The execution of Algorithm~\ref{algorithm:new_cont} is composed of three-phase rounds. At the first phase ``proposal'' of any round, every node sends a transfer proposal to one of its neighbors with the smallest load, if any. At the second phase ``deal'' of any round, every node accepts a single proposal sent to it, if any. 
At the third phase ``summary'', each node $u$ sends the value of its updated load to all its neighbors, resulting from the agreed transfers to and from it, if any. During the execution of the algorithm, we break the ties arbitrarily.

The rest of this sub-section is devoted to the analysis of Algorithm~\ref{algorithm:new_cont}. 
Let $L_{avg}$ be the average value of $load$ over $V$; note that this value does not change as a result of load transfers between nodes. Let us introduce \emph{potentials} 
as follows. We define $p(u)=(load(u)-L_{avg})^2$ be the potential of node $u$, and $p(G) = \sum_{u \in V} p(u)$ be the potential of $G$. 
Let us call a transfer of load $l$ from node $u$ to its neighbor $v$ \emph{fair} if $load(u)-load(v) \geq 2l$.

\begin{lemma}
	\label{l:fair}
Any fair transfer of load $l$ decreases the graph potential by at least $2l^2$.
\end{lemma}

\begin{proof}
Consider the state before a fair transfer of load $l$ from $u$ to $v$. Denote $a=load(v)$, $a' = load(u) \geq a+2l$. 
	The potential decrease by the transfer will be
	$ (a'^2 - (a'-l)^2) + (a^2 - (a+l)^2) = l (2a'-l) - l(2a+l) $ 
	$= l (2(a'-a)-2l) \geq 2l (2l-l) = 2l^2, $
	as required. 
\end{proof}

Note that all transfers in Algorithm~\ref{algorithm:new_cont} are fair. Hence the potential of graph $G$ never increases, implying that its decreases only accumulate.

\begin{lemma}
	\label{l:round_decrease}
	If the discrepancy of $G$ at the beginning of some round is $K$, the potential of $G$ decreases after that round by at least $K^2/2D$.
\end{lemma}

\begin{proof}
\begin{sloppypar*}
	Consider an arbitrary round. Let $x$ and $y$ be nodes with load $L_{max}$ and $L_{min}$, respectively, and let $P$ be a \emph{shortest} path from $y$ to $x$, $P=(y=v_0, v_1, v_2, \dots, v_k=x)$. Note that $k \le D$. Consider the sequence of edges $(v_{i-1}, v_i)$ along $P$, and choose its sub-sequence $S$ consisting of all edges with $\delta_i = load(v_i) - load(v_{i-1}) > 0$.  Let $S = (e_1=(v_{i_1 -1}, v_{i_1}), e_2=(v_{i_2 -1}, v_{i_2}), \dots,  e_{k'}=(v_{i_{k'} -1}, v_{i_{k'}}))$, $k' \le k \le D$. 
	Observe that by the definition of $S$, interval $[L_{min}, L_{max}]$ on the load axis is covered by intervals $[load(v_{i_j - 1}), load(v_{i_{j-1}})]$, since 
	\emph{$load(v_{i_1 -1})=L_{min}$, $load(v_{i_{k'}})=L_{max}$, and for any $2 \le j \le k'$, $load(v_{i_{j-1}}) \ge load(v_{i_j - 1})$} (we call this $property~1$). As a consequence, the sum of load differences $\sum_{j=1}^{k'} \delta_{i_j}$ over $S$ is at least $L_{max} - L_{min} = K$. 
	
	Since for every node $v_{i_j}$, its neighbor $v_{i_{j}-1}$ has a strictly lesser load, the condition of the first \textbf{if} in Algorithm~\ref{algorithm:new_cont} is satisfied for each $v_{i_{j}}$. Thus, each $v_{i_{j}}$ proposes a transfer to its minimally loaded neighbor; denote that neighbor by $w_j$. Note that the transfer amount in that proposal is at least $\delta_{i_j}/2$. 
	Hence, \emph{the sum of load proposals issued by the heads of edges in $S$ is at least $K/2$} (we call this $property~2$).
	By the algorithm, each node $w_i$ accepts the biggest proposal sent to it, which value is at least $\delta_{i_j}/2$. 
	
	Consider the simple case when all nodes $w_j$ are different. Then by Lemma~\ref{l:fair}, the total decrease of the potential at the round, $\Delta$, is at least $\sum_j 2(\delta_{i_j}/2)^2$. 
	By simple algebra, for a set of at most $D$ numbers with a sum bounded by $K$, the sum of numbers' squares is minimal if there are exactly $D$ equal numbers summing to $K$. We obtain $\Delta \ge D \cdot 2 (K/2D)^2 = K^2/2D$, as required.
	
	Let us now reduce the general case to the simple case as above. Suppose that several nodes $v_{i_{j}}$ proposed a transfer to same node $w$. Denote by $v_{i_{j'}}$ the first such node along $P$, by $v_{i_{j''}}$ the last one, and by $v_{i_{\bar j}}$ the node with the maximal load among them. 
	Let us shorten $S$ by replacing its sub-sequence from $e_{j'}$ to $e_{j''}$ by the single edge $\bar e=(w, v_{i_{\bar j}})$, with $\bar \delta = load(v_{i_{\bar j}})-load(w) > 0$; we denote the result of the shortcut by $\bar S$. 
	Let us show that $\bar S$ obeys property~1.
	Since  $v_{i_{j'}}$ proposed to $w$, we have $load(v_{i_{j'}-1}) \ge load(w)$. By property~1 for $S$, $load(v_{i_{j'-1}}) \ge load(v_{i_{j'} - 1})$, and hence $load(v_{i_{j'-1}}) \ge load(w)$, as required.
	By the choice of $v_{i_{\bar j}}$, we have $load(v_{i_{\bar j}}) \ge load(v_{i_{j''}})$. By property~1 for $S$, $load(v_{i_{j''}}) \ge load(v_{i_{j''+1} - 1})$, and hence $load(v_{i_{\bar j}}) \ge load(v_{i_{j''+1} - 1})$, as required.
	As in the analysis of $S$, the sum of load differences along $\bar S$ is at least $K$, and the transfer value in the proposal of the head of each edge in $\bar S$ is at least a half of the load difference along that edge; Property~2 for $\bar S$ follows. Now, node $w$ gets a proposal from $v_{i_{\bar j}}$ only, among the heads of edges in $\bar S$. If there are more nodes $w_j$ with several proposals from heads of edges in $\bar S$, we make further similar shortcuts of $\bar S$. In this way, we eventually arrive at the simple case obeying both properties~1 and 2, which suffices. 
\end{sloppypar*}
\end{proof}

\begin{lemma}
	\label{l:total_decrease}
	For any positive $\beta$, after at most $(2nD+1) \ln(\lceil nK^2/\beta \rceil)$ rounds, the potential of $G$ will permanently be at most $\beta$.
\end{lemma}

\begin{proof}
	Consider any state of $G$, with discrepancy $K$. The potential of each node is at most $K^2$; hence, the potential of $G$ is at most $n K^2$. By Lemma~\ref{l:round_decrease}, the relative decrease of $p(G)$ at any round is at least by a factor of $1-\frac{K^2/2D}{n K^2} = 1-\frac 1{2n D}$. Therefore, after $2nD+1$ rounds, the relative decrease of $p(G)$ will be at least by a factor of $e$. Hence, after $(2nD+1) \ln(\lceil nK^2/\beta \rceil)$ rounds, beginning from the initial state, the potential will decrease by a factor of at least $\lceil nK^2/\beta \rceil$. Thus, it becomes at most $\beta$ and will never increase, as required.
\end{proof}

\begin{lemma}
	\label{l:monotonic}
	As the result of any single round of Algorithm~\ref{algorithm:new_cont}$:$ 
	\begin{enumerate}
	    \item If node $u$ transferred load to node $v$, $load(v)$ does not become greater than $load(u)$ at the end of round.
	    \item No node load strictly decreases to $L_{min}$ or below and no node load strictly increases to $L_{max}$ or above. Therefore, the algorithm is monotonic.
	    \item The load of at least one node with load $L_{min}$ strictly increases and the load of at least one node with load $L_{max}$ strictly decreases.
	\end{enumerate}
\end{lemma}
\begin{proof}
	Consider an arbitrary round. By the description of a round, each node $u$ participates in at most two load transfers:  one where $u$ receives a load and one where it transfers a load. Let $u$ transfer load $d>0$ to $v$. Since the transfer is fair, that is $load(u)-load(v) \ge 2d$, the resulting load of $u$ is not lower than that of $v$. The additional effect of other transfers to $u$ and from $v$ at the round, if any, can only increase $load(u)$ and decrease $load(v)$, thus proving the first statement of Lemma.
	
	By the same reason of transfer safety, the new load of $u$ is strictly greater than the old load of $v$, which is at least $L_{min}$. The analysis of load of $v$ and $L_{max}$ is symmetric. We thus arrive at the second statement of Lemma.
	
	For the third statement, consider nodes $x$ and $y$ as defined at the beginning of proof of Lemma~\ref{l:round_decrease}. Denote the node that $x$ proposes a transfer to by $w$. Let node $w$ accept the proposal of node $z$ (maybe $z \neq x$). By the choosing rule at $w$, $load(z) \ge load(x) = L_{max}$, that is $load(z) = L_{max}$. Note that the load of $z$ strictly decreases after the transfer is accepted by $w$. Moreover, no node transfers load to $z$, since the load of $z$ is highest among all of the nodes.
	We consider the load changes related to $y$ similarly, thus proving the third statement.
\end{proof}

\begin{theorem}
\begin{sloppypar*}
	\label{t:algorithm_cont}
	Algorithm~\ref{algorithm:new_cont} is monotonic.
	After  at most $(6n+3) D \ln(\lceil nK^2/(\epsilon^2/2) \rceil) = O(nD \log(nK/\epsilon))$ time of its execution, the discrepancy of $G$ will be at most $\epsilon$.
\end{sloppypar*}	
\end{theorem}

\begin{proof}
\begin{sloppypar*}
Algorithm~\ref{algorithm:new_cont} is monotonic by Lemma~\ref{l:monotonic}. By Lemma~\ref{l:total_decrease}, after  $(6n+3) D \ln(\lceil nK^2/(\epsilon^2/2) \rceil) = O(nD \log(nK/\epsilon))$ time, the potential of $G$ will become permanently be at most $\epsilon^2/2$. In such a state, let $u$ and $v$ be nodes with load $L_{max}$ and $L_{min}$, respectively. Then, the current potential of $G$ will be at least $p(u)+p(v) = (L_{max}-L_{avg})^2 + (L_{min}-L_{avg})^2 \ge 2 (K/2)^2$. Thus, $K^2/2 \le p(u)+p(v) \le p(G) \le \epsilon^2/2$, and the desired bound $K \le \epsilon$ follows.
\end{sloppypar*}
\end{proof}

\subsection{Discrete Algorithm and Analysis}

\begin{algorithm}
\caption{Single Proposal: Discrete Version}
\label{algorithm:new_disc}
\KwIn{An undirected graph $G = (V,E, load)$}

\KwOut{Graph $G$ in a $1$-Balanced state}

         \SetKwFunction{FMain}{}
         \SetKwProg{Fn} {Execute forever} { do}{}
         \DontPrintSemicolon
          \Fn{}{
                \For{every node $u$} { 
                \If {$u$ has at least one neighbor with a load lesser than its own load by at least 2}
                {
                Find the first neighbor in order, $v$, with the maximal difference $load(u) - load(v)$\\
                Node $u$ sends to $v$ a transfer proposal of $p_{uv} = \lfloor (load(u) - load(v))/2 \rfloor$\\
                }
                }
                \For{every node $u$} {
                \If{there is at least one transfer proposal to $u$}{
                Find a neighbor, $w$, proposing to $u$ the transfer of the maximum value, $p_{wu}$\\
                Node $u$ makes a deal: increases its load by $p_{wu}$ and informs node $w$ on accepting its proposal\\
                }
                }
                \For{every node $u$} {
                Node $u$ updates its load w.r.t.\ the deal issued by it and the deal made on its proposal, if any, and sends the current value of $load(u)$ to every its neighbor\\
                }
            }
        
\end{algorithm}

Consider Algorithm~\ref{algorithm:new_disc}. All transfers are fair. Hence, the potential of $G$ never increases.
Algorithm~\ref{algorithm:new_disc} has a structure similar to that of Algorithm~\ref{algorithm:new_cont}. However, the discrete nature of integer loads implies that the final state should be 1-Balanced, and that transfer values should be rounded. Accordingly, the algorithm analysis is different, though based partly on the same ideas and techniques as those in the analysis of Algorithm~\ref{algorithm:new_cont}.

The rest of this sub-section is devoted to the analysis of Algorithm~\ref{algorithm:new_disc}. 
Note that there is at least one transfer proposal at some round of Algorithm~\ref{algorithm:new_disc} if and only if the state of $G$ is not 1-Balanced. Hence, whenever a 1-Balanced state is reached, it is preserved forever. 

\begin{lemma}
	\label{l:round_decrease_disc}
	If the discrepancy of $G$ at the beginning of some round is $K \ge 2D$, the potential of $G$ decreases after that round by at least $K^2/8D$.
\end{lemma}

\begin{proof}
\begin{sloppypar*}
	In this proof, we omit some details similar to those in the proof of Lemma~\ref{l:round_decrease}. Consider an arbitrary round. Let $x$ and $y$ be nodes with load $L_{max}$ and $L_{min}$, respectively, and let P be a shortest path from $y$ to $x$, $P=(y=v_0, v_1, v_2, \dots, v_k=x)$. Note that $k \le D$. Consider the sequence of edges $(v_{i-1}, v_i)$ along $P$, and choose its sub-sequence $S$ consisting of all edges with $\delta_i = load(v_i) - load(v_{i-1}) \ge 1$.  Let $S = (e_1=(v_{i_1 -1}, v_{i_1}), e_2=(v_{i_2 -1}, v_{i_2}), \dots,  e_{k'}=(v_{i_{k'} -1}, v_{i_{k'}}))$, $k' \le k \le D$. Observe that by the definition of $S$, holds: $load(v_{i_1 -1})=L_{min}$, $load(v_{i_{k'}})=L_{max}$, and for any $1 \le j \le k'$, $load(v_{i_{j-1}}) \ge load(v_{i_j - 1})$. As a consequence, $\sum_{j=1}^{k'} \delta_{i_j} \ge L_{max} - L_{min} = K$.
	The sum of load proposals of all $v_{i_j}$ is $\sum_{j=1}^{k'} \lfloor \delta_{i_j}/2 \rfloor \ge  \sum_{j=1}^{k'} (\delta_{i_j}-1)/2 = (\sum_{j=1}^{k'} \delta_{i_j} - D)/2  \ge (K-D)/2$; by the condition $K \ge 2D$, this sum is at least $K/4$. (Note that for any node $v_{i_j}$ with $\delta_{i_j}=1$, it may happen that it made no proposal. Accordingly, we accounted the contribution of each such node $v_{i_j}$ to the above sum by $\lfloor \delta_{i_j}/2 \rfloor = 0$.)
	
	Each node $v_{i_{j}}$ with $\delta_{i_j} \ge 2$ proposes a transfer to its minimally loaded neighbor; denote it by $w_j$. Note that the transfer amount in that proposal is at least $\lfloor \delta_{i_j}/2 \rfloor$. By the algorithm, each node $w_i$ accepts the biggest proposal sent to it, of amount of at least $\lfloor \delta_{i_j}/2 \rfloor$.  
	
	Consider the simple case when all nodes $w_j$ are different. Then by Lemma~\ref{l:fair}, the total potential decrease at the round, $\Delta$, is at least $2 \sum_j \lfloor \delta_{i_j}/2 \rfloor ^2$. By simple algebra, for a set of numbers with a fixed sum and a bounded amount of numbers, the sum of numbers' squares is minimal if the quantity of those numbers is maximal and the numbers are equal. In our case, we obtain $\Delta \ge 2(D  (K/4D)^2) = K^2/8D$, as required.
	
	The reduction of the general case to the simple case is as in the proof of Lemma~\ref{l:round_decrease}.
\end{sloppypar*}
\end{proof}

\begin{lemma}
	\label{l:monotonic_disc}
	As the result of any single round of Algorithm~\ref{algorithm:new_disc}$:$ 
	\begin{enumerate}
	    \item If node $u$ transferred load to node $v$, $load(v)$ does not become greater than $load(u)$ at the end of round.
	    \item No node load strictly decreases to $L_{min}$ or below and no node load strictly increases to $L_{max}$ or above. Therefore, the algorithm is monotonic.
	\end{enumerate}
\end{lemma}

The proof is similar to that of Lemma~\ref{l:monotonic}.

\begin{lemma}
	\label{l:total_decrease_disc}
	After at most $(8nD+1) \ln(\lceil nK^2/(2D^2) \rceil)$ rounds, the discrepancy will permanently be less than $2D$.
\end{lemma}
\begin{proof}
    Assume to the contrary that the statement of Lemma does not hold.
    By Lemma~\ref{l:monotonic_disc}, Algorithm~\ref{algorithm:new_disc} is monotonic, and hence the discrepancy never increases. Therefore, during the first $(8nD+1) \ln(\lceil nK^2/(2D^2) \rceil)$ rounds of the algorithm, the discrepancy permanently is at least $2D$.
    Note that since the discrepancy is at least $2D$, the potential of $G$ is at least $2 D^2$.
    
	Consider any state of $G$, with discrepancy $K \ge 2D$. The potential of each node is at most $K^2$; hence, $p(G) \le n K^2$. By Lemma~\ref{l:round_decrease_disc}, the relative decrease of $p(G)$ at any round is at least by a factor of $1-\frac{K^2/8D}{n K^2} = 1-\frac 1{8n D}$. Therefore, after $8nD+1$ rounds, the relative decrease of $p(G)$ will be more than by a factor of $e$. Hence after $(8nD+1) \ln(\lceil nK^2/(2D^2) \rceil)$ rounds, beginning from the initial state, the potential will decrease by a factor of more than $\lceil nK^2/(2D^2) \rceil$. Thus, it becomes less than $2D^2$, a contradiction.
\end{proof}

\begin{theorem}
	\label{t:algorithm_disc}
	\begin{sloppypar*}
	Algorithm~\ref{algorithm:new_disc} is monotonic.
	After at most $(24n+3) D \ln(\lceil nK^2/(D^2/2) \rceil) + 6n D^2 = O(nD \log(nK/D) + n D^2)$ time of its execution, the graph state will become 1-Balanced and fixed.
	\end{sloppypar*}
\end{theorem}
\begin{proof}
Algorithm~\ref{algorithm:new_disc} is monotonic by Lemma~\ref{l:monotonic_disc}. By Lemma~\ref{l:total_decrease_disc}, after  $(24n+3) D \ln(\lceil nK^2/(2D^2) \rceil) = O(nD \log(nK/D))$ time, the discrepancy will become less than $2D$. In such a state, the potential $p(G)$ is less than $n(2D)^2=4n D^2$. At each round before arriving at a 1-Balanced state, at least one transfer is executed, and thus the potential decreases by at least 2. Hence after at most $2n D^2$ additional rounds in not 1-Balanced states, the potential will vanish, which means the fully balanced state. Summarizing, a 1-Balanced state will be reached after the time as in the statement of Theorem. By the description of the algorithm, it will not change after that.
\end{proof}

\emph{Remark\/}:
We believe that the running time bounds of deal-agreement based distributed algorithms for load balancing could be improved by future research. This is since up to now, we used only a restricted set of tools: the current bounds of Algorithm~\ref{algorithm:new_cont} and the first phase of Algorithm~\ref{algorithm:new_disc} are based on an analysis of a single path in the graph at each iteration, while the bound for the second phase of Algorithm~\ref{algorithm:new_disc} is based on an analysis of a single load transfer at each iteration.


\section{Multi-Neighbor Load Balancing Algorithm} 
\label{s:diffusion}

In this section, we present another monotonic synchronous load balancing algorithm, which is based on \emph{distributed proposals}. There, each node may propose load transfers to several of its neighbors, aiming to equalize the loads in its neighborhood as much as possible. \remove{\textcolor{purple}{Intuitively, this should speed up the convergence, as compared with single proposal algorithms.}} This variation, exchanging loads in parallel, is expected to speed up the convergence, as compared with single proposal algorithms. We formalize this as follows.
 Consider node $p$ and the part $\mathcal{V}_{less}(p)$ of its neighbors with loads smaller than $load(p)$. The plan of $p$ is to propose to nodes in $\mathcal{V}_{less}(p)$ in such a way that if all its proposals would be accepted, then the resulting minimal load in the node set $\mathcal{V}_{less}(p) \cup \{p\}$ will be maximal. 
 Note that for this, the resulting loads of $p$ and of all nodes that it transfers loads to should become equal. (Compare with the scenario, where we pour water into a basin with unequal heights at its bottom: the water surface will be flat.) Such proposals can be planned as follows.
 
\begin{algorithm}[!b]
\caption{Distributed Proposal: Multi-Neighbor Load Balancing Algorithm}
\label{algo:a1}
\KwIn{An undirected graph $G = (V,E, load)$}
\KwOut{Graph $G$ in a $1$-Balanced state}


Each node $p$ repeatedly executes line 3 to 19 and 20 to 29 concurrently\\

\SetKwFunction{FMain}{}
         \SetKwProg{Fn}{Upon a PULSE execute}{}{}
         \DontPrintSemicolon
            \Fn{\FMain}{
                Node $p$ reads the load of its neighbors\\
                Compute $\mathcal{V}_{less}$ = $\mathcal{V}_{less}(p)$, $\mathcal{V}_{more}$ = $\mathcal{V}_{more}(p)$ \\
                
                \If{$q = |\mathcal{V}_{less}| > 0$}{ 
                    $TentativeLoad = load(p)$ \\
                    Sort and number $\mathcal{V}_{less}$ in a non-decreasing fashion as $p_1, p_2, \dots p_q$ \\
                    \For {$i=1$ to $q$} {$prop(i)=0$}
                    $i = 1$ \\
                    
                     \While{$TentativeLoad \ge load(p_i) + prop(i) + 2$}{
                        $TentativeLoad = TentativeLoad - 1$ \\
                        $prop(i) = prop(i) + 1$ \\
                        \If{$i<q \land load(p_i)+prop(i) > load(p_{i+1}) + prop(i+1)$} 
                            {$i = i + 1$} 
                        \Else {$i = 1$} 
                        }
                        
                    \For {$i=1$ to $q$} 
                         {Send a proposal ($prop(i), TentativeLoad$) to $p_i$}
                        
                }
                 \SetKwFunction{FMain}{}
                 \SetKwProg{Fn}{Upon arrival of proposal}{(ProposeToTransfer, TentativeLoad)}{}
                 \DontPrintSemicolon
                 \Fn{\FMain}{ 
                        $MaxLoad$ = Load of Maximum loaded node from $\mathcal{V}_{more}$\\
            
                        \textit{LoadToReceive} = $MaxLoad - load(p) - 1$\\
                      
                        Sort $\mathcal{V}_{more}$ in non-increasing fashion\\
                        
                        \While{$|\mathcal{V}_{more}| > 0 \land$ \textit{LoadToReceive} $> 0$}{
                        \textit{ProposeToReceive} = $\lfloor$ \textit{LoadToReceive}$/|\mathcal{V}_{more}| \rfloor$\\
                         
                         \textbf{Final Deal:} Receive the additional load \textit{ProposeToReceive} from each node of $\mathcal{V}_{more}$ in Round-Robin fashion:\\

                         \nonl \If{Load of node from $\mathcal{V}_{more}$ == \textit{TentativeLoad}}{
                         \nonl Remove node from $\mathcal{V}_{more}$ 
                         }
                       
                        $Deal$ = \textit{min(ProposeToTransfer, ProposeToReceive)}\\
                        $load(p) = load(p) + Deal$\\
                        $load(q) = load(q) - Deal$
                                     
                        } 
                 }
            }
 
\end{algorithm} 
 Let us sort the nodes in $\mathcal{V}_{less}(p)$ from the minimal to the maximal load. Consider an arbitrary prefix of that order, ending by node $r$; denote by $avg_r$ the average of loads of $p$ and the nodes in that prefix.
 Let the maximal prefix such that $avg_r > load(r)$ be with $r=r^*$; we denote the node set in that prefix by $\mathcal{V}'_{less}(p)$. Node $p$ proposes to each node $q \in \mathcal{V}'_{less}(p)$ a transfer of value $avg_{r^*} - load(q) > 0$. If all those proposals would be accepted, then the resulting loads of $p$ and of all nodes in $\mathcal{V}'_{less}(p)$ will equal $avg_{r^*}$.
 In the discrete problem setting, some of those proposals should be $\lfloor avg_{r^*} - load(q) \rfloor$ and the others be $\lceil avg_{r^*} - load(q) \rceil$, in order to make all proposal values integers; then, the planned resulting loads at the nodes in $\mathcal{V}_{less}(p) \cup \{p\}$ will be equal up to 1.

Let us show advantages of the distributed proposal approach on two generic examples. 
 Consider the following ``lucky'' example, where the running time of a distributed proposal algorithm is $O(1)$. Let the graph consist of $n$ nodes: $p$ and $n-1$ its neighbors, with an \emph{arbitrary set of edges} between the neighbors. Let $load(p)=n^2$ and the loads of its neighbors be \emph{arbitrary integers} between 0 and $n$. Node $p$ will propose to all its neighbors. If the strategy of proposal acceptance would be by preferring the maximal transfer suggested to it, then all proposals of $p$ will be accepted, and thus the resulting loads in the entire graph \emph{will equalize after a single round} of proposal/acceptance.

 As another example motivating distributed proposals, let us consider the on-line setting where the discrepancy is maintained to be usually small, but large new portions of load can arrive at some nodes at unpredictable moments. Let load $L >> L_{max}$ arrive at node $p$ with $r$ neighbors. The next round will equalize the loads at $p$ and all its neighbors. Then, the maximal load will become at most $L_{max}+(L-L_{max})/((r+1)$, that is the initial discrepancy jump by at least $L-L_{max}$ will be provably decreased by a factor of $r+1$ in time $O(1)$.

Algorithm~\ref{algo:a1} implements the described distributed proposal approach for the discrete problem setting. The proposal acceptance strategy in it is also based on a similar neighborhood-equalizing idea. 
The set of neighbors of node $p$ with loads greater than $load(p)$ is denoted there by $\mathcal{V}_{more}(p)$.
We believe that the distributed proposal algorithms as a generalization of single proposal ones benefit from the same complexities. 
Though, 
we present a simple analysis which assists us in the (proof-based) design of the algorithms.

\begin{lemma}
\label{lemma:lemma1}
As the result of any round of Algorithm \ref{algo:a1}, the maximum individual load does not increase and the minimum individual load does not decrease.
\end{lemma}
\begin{proof}
 Consider an arbitrary round, where $L_{max}$ holds the maximum individual load and $L_{min}$ holds the minimum individual load. When $|\mathcal{V}_{less}| \neq 0$ then node $p$ proposes a proposals to each node of $\mathcal{V}_{less}$  (Line 19) but if the load of any node from $\mathcal{V}_{less}$ exceeds \textit{TentativeLoad} then remove that node from $\mathcal{V}_{less}$. Similarly, when the nodes receive the proposal then node $p$ receives the deal from each node in $\mathcal{V}_{more}$  (Line 26) but if the load of any node from $\mathcal{V}_{more}$ reaching less than \textit{TentativeLoad} then remove that node from $\mathcal{V}_{more}$. Thus, the maximum individual load does not increase and the minimum individual load does not decrease.
\end{proof}

\noindent
Define the set of nodes with the minimum individual load:
\vspace{-0.1cm}
\begin{flushright}
$\mathcal{SL}_{min} = \{p \in \mathcal{N} | load(p)  = {L}_{min}\}$, where ${L}_{min} = \min\limits_{p \in \mathcal{N}}load(p)$\\
\end{flushright}

\noindent
Similarly, define the set of nodes with the maximum individual load:
\vspace{-0.1cm}
\begin{flushright}
$\mathcal{SL}_{max} = \{p \in \mathcal{N} | load(p)  = {L}_{max}\}$, where ${L}_{max} = \max\limits_{p \in \mathcal{N}}load(p)$
\end{flushright}

\begin{lemma}
\label{lemma:lemma2}
Algorithm \ref{algo:a1} guarantees that no node joins the sets  $\mathcal{SL}_{min}$ and $\mathcal{SL}_{max}$.
\end{lemma}
\begin{proof}{}
First we show that no node joins $\mathcal{SL}_{min}$. A node $p$ that does not belong to $\mathcal{SL}_{min}$ must give loads in order to join $\mathcal{SL}_{min}$. However, before proposing the load to send in each round node $p$ compares the \textit{TentativeLoad} with the load of each node in $\mathcal{V}_{less}$, So, Algorithm \ref{algo:a1} ensures that node $p$ never proposes to give load amount that makes his load less than $L_{min}$.

Analogously, we show that no node joins $\mathcal{SL}_{max}$.
A node $p$ that does not belong to $\mathcal{SL}_{max}$ must receive enough loads in order to join $\mathcal{SL}_{max}$. However, before receiving the load in each round node $p$ compares the \textit{TentativeLoad} with the load of each node in $\mathcal{V}_{more}$, So, Algorithm \ref{algo:a1} ensures that node $p$ never receives load amount that makes its load more than $L_{max}$. Hence, no node joins the sets  $\mathcal{SL}_{min}$ and $\mathcal{SL}_{max}$.
\end{proof}

\begin{corollary}
In any round, as long as the difference between any neighboring pair is not 1 or 0, then repeatedly, the size of $\mathcal{SL}_{min}$ and/or $\mathcal{SL}_{max}$ monotonically shrinks until the gap between $L_{max}$ and $L_{min}$ is reduced. Therefore the system convergences toward being \textit{1-Balanced} while the difference between any neighboring pair is greater than 1.
\end{corollary}
\begin{proof}
Lemma \ref{lemma:lemma2} establishes no node joins $\mathcal{SL}_{min}$ and $\mathcal{SL}_{max}$. Algorithm \ref{algo:a1} executes repeatedly until the deals happen, and load transfers from the higher loaded node to a lesser loaded node. So, a member of one of $\mathcal{SL}_{min}$ and $\mathcal{SL}_{max}$ sets leaves. Lemma \ref{lemma:lemma1} ensures that no lesser (higher) load value than the values in $\mathcal{SL}_{min}$ ($\mathcal{SL}_{max}$, respectively) is introduced. Once one of these sets becomes empty, a new set is defined instead of the empty set, implying a gap shrink between the values in the two sets.
\end{proof}

\begin{lemma}
\label{lemma:lemma3}
Algorithm \ref{algo:a1} guarantees  potential function converges after each load transfer. 
\end{lemma}

\begin{proof}
Consider an arbitrary round, where $load(p)$ represents the load node of $p$ and $L_{avg}$ computes the average load in the whole graph. $L_{avg}$ is the same in any round.
We consider a potential function for analyzing convergence of algorithm: $\sum (load(p) - L_{avg})^ 2$. 
Assumes node $i$ transferring 1 unit load to node $j$, where $L_{i} > L_{j}$ . Here we analyze that the potential function is decreasing after transferring the load from the higher load to the lower load after each deal. Potential function value before transferring the load:
$(load(i) - L_{avg})^ 2 + (load(j) - L_{avg} )^ 2$ . Potential function value after transferring the load: $(load(i) -1 - L_{avg})^ 2 + (load(j) + 1- L_{avg})^ 2$ . Potential function difference should shrink namely:
$ ((load(i) - L_{avg})^ 2 + (load(j) - L_{avg} )^ 2) - ((load(i) -1 - L_{avg})^ 2 + (load(j) +1- L_{avg})^ 2)  > 0$ 

After expansion:  $ load(i) > load(j) + 1$, Which follows our required $ load(i) > load(j) + 1$ condition for the algorithm to finalize a deal. Since $L_{avg}$ is fixed. According to Algorithm \ref{algo:a1}, condition ensures legitimate load transfer from the higher load to the lower load. As node $p$ receives the load same as $L_{avg}$ but condition $|\mathcal{V}_{less}| \neq 0$ satisfies, then node $p$ starts transferring the load to the less loaded node by which potential function converges.
\end{proof}

\begin{theorem}
Algorithm \ref{algo:a1} is monotonic. After $O(nK^2)$ time the initial discrepancy of $K$ will permanently be 1-Balanced, where $n$ is the total number of nodes in the graph.
\end{theorem}

\begin{proof}
Lemma \ref{lemma:lemma1} and \ref{lemma:lemma2} establish the monotonicity of Algorithm \ref{algo:a1}. Since at least one deal is executed in a constant number of message exchanges (read loads, proposals, deals) the algorithm takes time proportional to the total number of deals. Deals are executed until all nodes are \emph{1-Balanced}. Thus, if at least one deal is happening in $O(1)$ time then the algorithm will converge in $O(nK^2)$ time.
\end{proof}

\section{Asynchronous Load Balancing Algorithms}
\label{s:async}

\remove{
\textcolor{red}{(R3: The fourth algorithm is unnecessarily complicated, and the complicatedness does not contribute to the performance of the algorithm. For example, each time a node finds two or more less-loaded neighbors, it computes a subset of those neighbors and sends a part of its loads to them. However, it suffices to send a part of its load only to the minimum loaded neighbor. (Ties can be broken arbitrarily.) We can still get the same upper bound on the time complexity as Theorem 4, i.e., $O(nK^2)$ time.)}
}

In this section, we describe our main techniques for achieving asynchronous load balancing. We consider the undirected connected communication graph $G = (V, E, load)$, in which nodes $v_i \in V$ hold an arbitrary non-negative $load(v)$. Nodes are communicating using message passing along communication graph edges. Neighboring nodes in the graph can communicate by sending and receiving messages in FIFO order. 
In asynchronous systems, a message sent from a node to its neighbor eventually arrives at its destination. However, unlike in synchronous (and semi-synchronous) systems, there is no time-bound on the time it takes for a message to arrive at  its destination. Note that one may suggest using a synchronizer to apply a synchronous algorithm. Such an approach will slow down the load balancing activity to reflect the slowest participant in the system. 

We consider that standard settings of asynchronous systems, see e.g.,~\cite{10.5555/335041,10.5555/525656}. A configuration of asynchronous systems is described by a vector of states of the nodes and message (FIFO) queues, one queue for each edge. The message queue consists of all the messages sent, over the edge, and not yet received. System configuration is changed by an atomic step in which a message is sent or received (local computations are assumed to take negligible time). An atomic step in which a message is sent in line 26 of Algorithm \ref{algo:a4}, is called the atomic \textit{deal} step, or simply a deal. 

Our asynchronous load balancing algorithm is based on \emph{distributed proposals}. There, each node may propose load transfers to several of its neighbor by computing $\mathcal{PV}_{less}(p)$, which is part of $\mathcal{V}_{less}(p)$. $\mathcal{PV}_{less}(p)$ is the resulting minimal loaded node set whose load is less than \emph{TentativeLoad} after all proposals gets accepted. While sending the proposal, each node sends the value of \emph{LoadToTransfer} (load which can be transferred to neighboring node) and \emph{TentativeLoad} (load of the node after giving loads to its neighbors) with all set of nodes in $\mathcal{PV}_{less}(p)$. After receiving the proposal, the node sends an acknowledgment  to the sender node; the sender node waits for an acknowledgment  from all nodes of $\mathcal{PV}_{less}(p)$.

The asynchronous algorithm ensures that the local computation between two nodes is assumed to be before the second communication starts. Consider an example when a node $q$ of
$\mathcal{PV}_{less}$ receives a proposal, the deal happens between node $p$ and
node $q$. In this case \textit{TentativeLoad} of node $p$ is always greater than the load of node $q$ (when $q$ responds to the deal) because node $p$ is waiting for acknowledgments from all nodes of $\mathcal{PV}_{less}$. 

\begin{algorithm}[!b]
\caption{Asynchronous Load Balancing Algorithm }
\label{algo:a4}
\KwIn{An undirected graph $G = (V,E, load)$}
\KwOut{Graph $G$ in a 1-Balanced state}

$LastReceivedLoad = 0$\\
$LastGaveLoad = 0$\\

Each node $p$ repeatedly executes line 5 to 22, 23 to 30, and 31 to 33 concurrently  \\
         \SetKwFunction{FMain}{}
         \SetKwProg{Fn}{Execute forever}{  do}{}
         \DontPrintSemicolon
          \Fn{}{

                    $load(p) = load(p) +$ \textit{LastReceivedLoad} -\textit{ LastGaveLoad}\\
                    \textit{LastReceivedLoad}= 0\\
                    \textit{LastGaveLoad} = 0\\
                    $TLoad(p) = load(p)$\\
                    Node $p$ reads the load of its neighbors\\
                    Compute $\mathcal{V}_{less}$ = $\mathcal{V}_{less}(p)$ \\
                   
                      \If{$ |\mathcal{V}_{less}| \neq 0$}{    
                              $MinLoad$ = Load of minimum loaded node from $\mathcal{V}_{less}$\\
                              \textit{LoadToTransfer} = $ \lfloor(TLoad(p) - MinLoad)/2 \rfloor $\\
                              \textit{TentativeLoad} = $ TLoad(p)$ - \textit{LoadToTransfer} \\
                            
                             \SetKwFunction{FMain}{}
                             \SetKwProg{Fn}{for every $q \in \mathcal{V}_{less}$}{  do}{}
                             \DontPrintSemicolon
                             \Fn{}{
                                 \If{$load(q) <$ \textit{ TentativeLoad}}{
                                    $\mathcal{PV}_{less} = \mathcal{PV}_{less} \cup \{ q \}$ 
                                 }
                             }
                             $P_p$ = \textit{  RRProposal(LoadToTransfer, }$\mathcal{PV}_{less},$ \textit{TentativeLoad)} \\

                             \SetKwFunction{FMain}{}
                             \SetKwProg{Fn}{For every $q \in \mathcal{PV}_{less}$}{  do}{}
                             \DontPrintSemicolon
                             \Fn{}{
                             
                             Send ($P_p[q]$, \textit{TentativeLoad}) \\
                             $Ack_q = False$ \\
                             }

                             Wait until $Ack_q = True$ for every $q$ in $\mathcal{PV}_{less}$,
                            
                        } 
 }

         \SetKwFunction{FMain}{}
         \SetKwProg{Fn}{upon arrival of}{ from neighbor $q$  do}{}
         \DontPrintSemicolon
            \Fn{\FMain{$Proposal_q, TentativeLoad_q$}}{
            \If{$TentativeLoad_q - TLoad(p) > 0$
            }{
            $ Deal = min( (TentativeLoad_q - TLoad(p)), Proposal_q) $\\
            send to $q$ $AckMsg = Deal$\\
            
            \textit{LastReceivedLoad} = \textit{LastReceivedLoad} + \textit{Deal} \\
            $TLaod(p)$ = $TLoad(p)$ + \textit{Deal}
            }
            \Else{
            send to $q$ $AckMsg = 0$
            }
            }
      
        \SetKwFunction{FMain}{}
                    \SetKwProg{Fn}{upon}{ do}{}
                    \DontPrintSemicolon
                        \Fn{\FMain{$AckMsg \ reception \ from \ q $}}{
                        \textit{LastGaveLoad} =  \textit{LastGaveLoad} + \textit{AckMsg.Deal}\\
                         $Ack_q = True$
                         }  
                  
         \end{algorithm}
\setlength{\interspacetitleruled}{0pt}%
\setlength{\algotitleheightrule}{0pt}%
\begin{algorithm}[t]
\LinesNumbered
\setcounter{AlgoLine}{33}

        \SetKwFunction{FMain}{}
         \SetKwProg{Fn}{Procedure RRProposal}{}{}
         \DontPrintSemicolon
             \Fn{\FMain{\textit{LoadToTransfer}, $\mathcal{PV}_{less}$, \textit{TentativeLoad}}}{
             $\mathcal{TV}_{less} = \mathcal{PV}_{less}$ \\
             \textit{LeftLoadToTransfer} = \textit{LoadToTransfer}
             
             \While{$|\mathcal{TV}_{less}| > 0 \land$ \textit{LeftLoadToTransfer} $> 0$}{

                $m$ = $max (\mathcal{TV}_{less})$\\
                    \If{$ (TentativeLoad - m)\times |\mathcal{TV}_{less}|  \leq$ \textit{LeftLoadToTransfer}}{
                    
                    Update $P_p$ to propose to transfer additional load (\textit{TentativeLoad}-\textit{m}) to every node in $\mathcal{TV}_{less}$, and subtract from \textit{LeftLoadToTransfer}: \\
                   \nonl \If{Load of node from $\mathcal{TV}_{less}$ == \textit{TentativeLoad}}{
                    \nonl Remove node from $\mathcal{TV}_{less}$
                     }
                    
                    }

                \Else{
                    
                    Update $P_p$ to propose to transfer additional loads \textit{LeftLoadToTransfer} to node of $\mathcal{TV}_{less}$ in Round-Robin fashion, and subtract from \textit{LeftLoadToTransfer}
                }
             }
              \KwRet $P_p$   
              }

                  \end{algorithm}

Execution of Algorithm~\ref{algo:a4} is as follows: Every time each node makes copy of $load(p)$ in $TLoad(p)$, reads the load of its neighbors and computes $\mathcal{V}_{less}(p)$ and makes a copy in $\mathcal{V}_{less}$ . If $ |\mathcal{V}_{less}| \neq 0$ then each node computes \textit{MinLoad}, which stores the load of the minimum loaded node from $\mathcal{V}_{less}$. Computes \textit{LoadToTransfer} by computing $ \lfloor(TLoad(p) - MinLoad)/2 \rfloor $. Also computes \textit{TentativeLoad} by computing $TLoad(p) - LoadToTransfer$. For each node of $\mathcal{V}_{less}$ whose load is less than \textit{TentativeLoad} will be added into $\mathcal{PV}_{less}$.  After deciding the nodes in $\mathcal{PV}_{less}$ node $p$ sends proposal to each node of $\mathcal{PV}_{less}$ with proposal and \textit{TentativeLoad}.

Upon arrival of proposal each node individually checks $TentativeLoad_q - TLoad(p) > 0$ If satisfied, computes the \textit{Deal} by computing $ min( (TentativeLoad_q - TLoad(p)), Proposal_q)$. Sends the acknowledgement message as \textit{Deal} to each neighbor, updates the \textit{LastReceivedLoad} and $TLoad(p)$ by adding \textit{Deal} into them. Otherwise send \textit{0} as acknowledgement message. Node $p$ waits for acknowledgement from each node of $\mathcal{PV}_{less}$ and once it has received \textit{AckMsg} with the deal from its neighbors, updates the \textit{LastGaveLoad} and then node $p$ sets own acknowledgement \textit{True}.

The \textit{Round-Robin Proposal} starts by making a copy of  $\mathcal{PV}_{less}$ in  $\mathcal{TV}_{less}$ and \textit{LoadToTransfer} in \textit{LeftLoadToTransfer}. It  Keep updating proposal until $|\mathcal{TV}_{less}| > 0 \land LeftLoadToTransfer > 0$. If this condition satisfies, it store the maximum load of maximum loaded node of $\mathcal{TV}_{less}$ in \textit{m}. For each node of $\mathcal{TV}_{less}$ if $ (TentativeLoad - m)\times |\mathcal{TV}_{less}|  \leq LeftLoadToTransfer$ condition satisfied then it update $P_p$ to propose to transfer additional load $(TentativeLoad-m)$ to every node in $\mathcal{TV}_{less}$, and subtract from $LeftLoadToTransfer$. Any node from $\mathcal{TV}_{less}$ that has already received load equal to \textit{TentativeLoad} then remove that node from the $\mathcal{TV}_{less}$. If previous condition does not satisfy, then Update $P_p$ to propose to transfer additional loads \textit{LeftLoadToTransfer} to node of $\mathcal{TV}_{less}$ in the Round-Robin fashion, and subtract from $LeftLoadToTransfer$ and return the proposal $P_p$. Hence, node $p$ updates its own load by adding \textit{LastreceivedLoad} and subtracting \textit{LastGaveLoad}. As a result a deal completion happens.

During the execution of Algorithm \ref{algo:a4} each node repeatedly executes lines 5 to 22 (send proposal), 23 to 30 (upon arrival of the proposal), and 31 to 33 (upon acknowledgment reception) concurrently and forever.

\begin{lemma}
\label{l:async1}
In every deal the load transfer is from the higher loaded node to the lower loaded node.
\end{lemma}{}
\begin{proof}
We analyzed this using the interleaving model, in this model at the given
time only a single processor executes an atomic step. Each atomic step
consists of internal computation and single communication operation
(Send-Receive message). The atomic step may consist of local computations
(e.g., computation of load transfer between two nodes). The asynchronous algorithm
ensures that the local computation between two nodes is assumed to be before
the second communication starts. Consider an example when a node $q$ of
$\mathcal{PV}_{less}$ receives a proposal, the deal happens between node $p$ and
node $q$. In this case \textit{TentativeLoad} of node $p$ is always greater than the load of node $q$ because node $p$ is waiting for acknowledgments from all nodes of $\mathcal{PV}_{less}$.
\end{proof}{}

\begin{lemma}
\label{l:async2}
As a result of any round of Algorithm \ref{algo:a4}, the maximum individual load does not increase and the  minimum individual load does not decrease.
\end{lemma}
\begin{proof} Consider an arbitrary deal, in each deal node $p$ checks with the neighboring nodes ($\mathcal{V}_{less}$), those nodes whose load is less than \emph{TentativeLoad} will become part of $\mathcal{PV}_{less}$ and receive the proposal from node $p$. Upon the proposal arrival the node computes the \textit{Deal} and picks the minimum load among ($TentativeLoad_q - TLoad(p)$) and receives the proposal. This ensures that no node receives the additional load by which they exceed the maximum individual load, similarly, no node gives more loads, by which they retain less than the minimum individual load. 
\end{proof}

\begin{theorem} 
Algorithm \ref{algo:a4} is monotonic. After $O(nK^2)$ time the initial discrepancy of $K$ will permanently be 1-Balanced, where $n$ is the total number of nodes in the graph.
\end{theorem}
\begin{proof}
 Lemma \ref{l:async1} and \ref{l:async2} establish the monotonicity of algorithm \ref{algo:a4}.  
We now show that at least one deal is executed during a constant
number of messages exchanges. Assume towards contradiction that no
deal is executed, and hence, the loads are constant. The first reads in
these fixed loads execution must result in a correct value of loads,
and therefore followed by correct proposals and deals, hence, the
contradiction. Since at least one deal is executed in a constant number of message exchanges (read loads, proposals, deals) the algorithm takes time proportional to the total number of deals. Deals are executed until all nodes are \emph{1-Balanced}.  Thus, when at least one deal is happening in $O(1)$ time then the algorithm will converge in $O(nK^2)$ time.
 \end{proof}

\remove{A node repeatedly sends inquiring messages for loads of its neighbors, in short it reads the loads of its neighbors. }

\section{Self-Stabilizing Load Balancing Algorithms}
\label{s:selfstab}
\remove{
\textcolor{red}{(R1: Section 5 and App D, self-stabilizing algorithm:
It is not clear what is the correctness measure of the self-stabilizing
setting. It is possible "load(p)", "LastReceivedLoad" and
"LastGaveLoad" are corrupted and message channels may contain garbage
messages initially.
Do you mean that the final state (configuration) must be 1-Balancing
regardless the initial load(p) for each p?
For example, what will happen if "LastGaveLoad" has huge value
by corruption?
The discussion in this paper (App C), the authors simply says that
garbage messages are consumed and data link protocol eventually
recovers.
I do not think this discussion is enough.)}

\textcolor{red}{(R3:  I do not understand the correctness of the fifth algorithm. The claim of Theorem 5 itself is ambiguous. What do you mean by "the unknowns contents (duplicate and omitted) of the link are controlled"? What is "undesired messages"? Algorithm 5 is also unclear. This algorithm assumes the k-bounded channel, but k does not appear in the pseudocode.)}

\textcolor{red}{(R4: This paper fit in the topics of the conference but the self-stabilizing version of the algorithm is not very interesting. Indeed, this algorithm is obtained by the application of an existing self-stabilizing data-link protocol to the non stabilizing version of the algorithm and does not bring any new insight in the area of self-stabilization.)}}

In this section we present the first self-stabilizing load balancing algorithm. Our self-stabilizing load balancing algorithm  is designed for an asynchronous message passing system. The system setting for self-stabilizing load balancing algorithms is the same as the system settings in the asynchronous load balancing algorithm. The self-stabilization requirement to reach a suffix of the execution in a set of legal executions starting in an arbitrary configuration. Where execution is an alternating sequence of configurations and atomic steps, such that the atomic step next to a configuration is executed by a node according to its state in the configuration and the attached message queue \cite{10.5555/335041}. In the asynchronous system, deadlock can occur if the sender continuously waits for non-existing acknowledgments. Our Self-Stabilizing Load Balancing Algorithm uses the concept of Self-Stabilizing Data-Link with \emph{k-bounded} channel \cite{10.1007/boundedFIFO}, which is responsible for the eventual sending and receiving of loads, where in every deal load transfers from higher loaded node to lesser loaded node. 

\setlength{\interspacetitleruled}{1.5pt}%
\setlength{\algotitleheightrule}{0.8pt}%
\begin{algorithm}[!b]
\caption{Self-Stabilizing Load Balancing Algorithm}
\label{algo:ssLoad}
\KwIn{An undirected graph $G = (V,E, load)$}
\KwOut{Graph G in 1-Balanced state.}

$LastReceivedLoad = 0$\\
$LastGaveLoad = 0$\\

Each node repeatedly executes line 5 to 23, 24 to 31, and 32 to 34 concurrently  \\
         \SetKwFunction{FMain}{}
         \SetKwProg{Fn}{Repeat forever}{  do}{}
         \DontPrintSemicolon
          \Fn{}{
                   
                    $load(p) = load(p) + LastReceivedLoad - LastGaveLoad$\\
                    $LastReceivedLoad = 0$\\
                    $LastGaveLoad = 0$\\
                    $TLoad(p) = load(p)$\\

                     \SetKwFunction{FMain}{}
                     \SetKwProg{Fn}{ \textcolor{blue}{Repeat forever do}}{  }{}
                     \DontPrintSemicolon
                     \Fn{}{
                        Node $p$ reads the load of its neighbors\\
                         Compute $\mathcal{V}_{less}$ = $\mathcal{V}_{less}(p)$ \\
                     }

                      \If{$ |\mathcal{V}_{less}| \neq 0$}{    
                             \textit{MinLoad} = Load of minimum loaded node from $\mathcal{V}_{less}$\\
                             \textit{LoadToTransfer} = $ \lfloor(TLoad(p) - MinLoad)/2 \rfloor $\\
                              \textit{TentativeLoad} = $ TLoad(p) - LoadToTransfer $\\
                            
                             \SetKwFunction{FMain}{}
                             \SetKwProg{Fn}{for every $q \in \mathcal{V}_{less}$}{  do}{}
                             \DontPrintSemicolon
                             \Fn{}{
                                 \If{$load(q) < TentativeLoad$}{
                                    $\mathcal{PV}_{less} = \mathcal{PV}_{less} \cup load(q)$ 
                                 }
                             }
                             $P_p = RRProposal(LoadToTransfer,$ $\mathcal{PV}_{less}, TentativeLoad)$ \\

                             \SetKwFunction{FMain}{}
                             \SetKwProg{Fn}{Repeat forever} {  do}{}
                             \DontPrintSemicolon
                             \Fn{}{
                             
                             DataLinkSend $(P_p[q], TentativeLoad)$ \\
                             $Ack_q = False$ \\
                             
                             }

                            $Ack_q = True$ for every $q$ in $\mathcal{PV}_{less}$
                            
                            
                        } 
 }    
            
    \end{algorithm}
\setlength{\interspacetitleruled}{0pt}%
\setlength{\algotitleheightrule}{0pt}%
\begin{algorithm}[!t]
\LinesNumbered
\setcounter{AlgoLine}{23}       

         \SetKwFunction{FMain}{}
         \SetKwProg{Fn}{\textcolor{blue}{upon DataLinkArrival of }{\textcolor{blue}{  from neighbor $q$  do}}}{}
         \DontPrintSemicolon
            \Fn{\FMain{\textcolor{blue}{$Proposal_q, TentativeLoad_q$}}}{
            \If{$TentativeLoad_q - TLoad(p) > 0$
            }{
            $ Deal = min( (TentativeLoad_q - TLoad(p)), Proposal_q) $\\
             DataLinkSend to $q$ $AckMsg = Deal$\\
            
            $LastReceivedLoad = LastReceivedLoad + Deal$ \\
            $TLaod(p) = TLoad(p) + Deal$
            }
            \Else{
             \textcolor{blue}{DataLinkSend to $q$ $AckMsg = 0$}
            }
            }

        \SetKwFunction{FMain}{}
                    \SetKwProg{Fn}{\textcolor{blue}{upon}}{ \textcolor{blue}{do}}{}
                    \DontPrintSemicolon
                        \Fn{\FMain{\textcolor{blue}{$AckMsg \ DataLinkReception \ from \ q $}}}{
                        $LastGaveLoad =  LastGaveLoad + AckMsg.Deal$\\
                         $Ack_q = True$
                         }  

        \SetKwFunction{FMain}{}
         \SetKwProg{Fn}{Procedure RRProposal}{}{}
         \DontPrintSemicolon
             \Fn{\FMain{$LoadToTransfer$, $\mathcal{PV}_{less}$, $TentativeLoad$}}{
             $\mathcal{TV}_{less} = \mathcal{PV}_{less}$ \\
             $LeftLoadToTransfer$ = $LoadToTransfer$
             
             \While{$|\mathcal{TV}_{less}| > 0 \land LeftLoadToTransfer > 0$}{

                $m$ = $max (\mathcal{TV}_{less})$\\
                    \If{$ (TentativeLoad - m)\times |\mathcal{TV}_{less}|  \leq LeftLoadToTransfer$}{
                    
                    Update $P_p$ to propose to transfer additional load $(TentativeLoad-m)$ to every node in $\mathcal{TV}_{less}$, and subtract from $LeftLoadToTransfer$: \\
                   \nonl \If{Load of node from $\mathcal{TV}_{less}$ == $TentativeLoad$}{
                    \nonl Remove node from $\mathcal{TV}_{less}$
                     }
                    
                    }

                \Else{
                    
                    Update $P_p$ to propose to transfer additional loads \textit{LeftLoadToTransfer} to node of $\mathcal{TV}_{less}$ in Round-Robin fashion, and subtract from $LeftLoadToTransfer$
                }
             }
              \KwRet $P_p$   
              }

\end{algorithm}

 \remove{
\setlength{\interspacetitleruled}{1.5pt}%
\setlength{\algotitleheightrule}{0.8pt}%
\begin{algorithm}[H]
\label{algo:receiver}
\caption{Pseudo code for delivering message}
\SetKwFunction{FMain}{}
         \SetKwProg{Fn}{upon message arrival}{  do}{}
         \DontPrintSemicolon
          \Fn{}{
          Send ACK to sender
          }
          
         \SetKwProg{Fn}{upon \textit{last bit} received}{  do}{}
         \DontPrintSemicolon
          \Fn{}{
          \If{(last bit received == 0 $\land$ current bit received == 1)}{
          Deliver the message\\
          last bit received  = current bit received\\
          }
          }
\end{algorithm} 
}

 Here we use the concept of Self-Stabilizing Data-Link with \emph{k-bounded} channel, which is responsible for the eventual sending and receiving of loads. Starting in an arbitrary configuration with arbitrary messages in transit, a reliable data link eliminates the possibility of corrupted messages in the transient, and ensures that the actual load values are communicated among the neighbors. The retransmission of messages helps to avoid deadlocks, ensuring the arrival of an answer when waiting for an answer. In order to deliver the load, the sender repeatedly sends the message $\langle m_i, 0 \rangle$ to the receiver and the sender receives enough ACK from the receiver. 

The receiver sends ACK only when it receives a message from the sender. The sender waits to receive $2k+1$ ACKs before sending the next message $\langle m_i, 1 \rangle$. During this communication whenever the receiver identifies two consecutive messages $\langle m_i, 0 \rangle$ and then $\langle m_i, 1 \rangle$, the receiver delivers $m_i$ to the upper layer. Immediately after delivering a message, the receiver \emph{``cleans"} the possible corrupted incoming messages (including huge value of load by corruption) by ignoring the next \textit{k} messages.  Note that by the nature of transient faults the receiver may accept a phantom deal due to the arrival of a corrupted message that may not have originated from the sender. Still, following one such phantom deal, the links are cleaned and deals are based on actual load reports between neighbors and respects the higher load to lesser load invariant.
 
The system setting for self-stabilizing load balancing algorithms is the same as the system settings in the asynchronous load balancing algorithm. So additional changes in Algorithm \ref{algo:ssLoad} are highlighted in blue color in comparison to Algorithm~\ref{algo:a4}. Self-Stabilizing Data-Link algorithm ensures that message fetched by sender should be delivered by receiver without duplications and omissions. Similarly, \emph{DataLinkArrival}, \emph{DataLinkSend}, and \emph{DataLinkReception} ensure \emph{send}, \emph{arrival}, and \emph{reception} of message without duplications and omissions.

\begin{theorem} 
Algorithm~\ref{algo:ssLoad} ensures that the unknown contents (duplicate and omitted) of the link are controlled, and eliminates undesired messages from the link.
\end{theorem}
\begin{proof}
In asynchronous round of execution the node delivers duplicated or omitted message over k-bounded channel. Data-Link algorithm ensures, that the receiver node sends acknowledgement only when it receives message from sender node. Whenever the sender node sends  message to receiver node with either 0 or 1 bit, receiver node responds with an acknowledgement to the sender node. The sender sends messages in the alternate bit to the one that the receiver-node delivers the message from, swallowing of $k$ messages from transit after delivering the message ensures that the load always moves from the higher loaded node to lesser node. Thus, the new incoming loads transfer in without duplicate or omitted messages.
\end{proof}

\section{Conclusions}
\label{s:conclusion}
We have presented the class of local deal-agreement based load balancing algorithms and demonstrated a variety of monotonic anytime ones. Many 
details and possible extensions 
are omitted from this version. Still, we note that our scheme and concept can be easily extended e.g., to transfer loads directly up to a certain distance and/or to restrict the distance of load exchange with other nodes in the graph. The self-stabilizing solution can be extended to act as a super-stabilizing algorithm  \cite{DBLP:journals/cjtcs/DolevH97}, gracefully, dealing with dynamic settings, where nodes can join/leave the graph anytime, as well as handle loads received/dropped.

\bibliographystyle{splncs04}
\bibliography{ref}

\newpage
\appendix

\end{document}